%% file: arxiv.tex
\documentclass{article}
\usepackage{amsmath, amssymb,amsthm}
\usepackage{fullpage}
\usepackage{xyling, graphics}
\usepackage{tikz, fp}

\title{Unifying the Landscape of Cell-Probe Lower Bounds%
  \thanks{The conference version of this paper appeared in FOCS'08
    under the title \emph{(Data) Structures}.} }
\author{
Mihai P\v{a}tra\c{s}cu%
  \thanks{{\tt mip@alum.mit.edu}. AT\&T Labs.
	Parts of this work were done while the author was at MIT and IBM
        Almaden Research Center.} }

\newtheorem{theorem}{Theorem}
\newtheorem{lemma}[theorem]{Lemma}

\newtheorem{reduction}[theorem]{Reduction}
\newtheorem{claim}[theorem]{Claim}

\newcommand{\poly}{\mathrm{poly}}
\newcommand{\polylog}{\mathrm{polylog}}

\newcommand{\Omegat}{\widetilde{\Omega}}
\newcommand{\eps}{\varepsilon}

\newcommand{\calD}{{\cal D}}
\newcommand{\calS}{\mathcal{S}}
\newcommand{\calT}{\mathcal{T}}

\newcommand{\Dyes}{\calD_{\textnormal{\sc yes}}}
\newcommand{\coins}{{\cal C}}
\newcommand{\err}{{\cal E}}

\DeclareMathOperator*{\E}{\mathbb{E}}
\DeclareMathOperator*{\HH}{{\rm H}}
\DeclareMathOperator*{\I}{{\rm I}}

\begin{document}

\maketitle

\begin{abstract}
We show that a large fraction of the data-structure lower bounds known
today in fact follow by reduction from the communication complexity of
lopsided (asymmetric) set disjointness. This includes lower bounds
for:
\begin{itemize}
\item high-dimensional problems, where the goal is to show large space
  lower bounds.

\item constant-dimensional geometric problems, where the goal is to
  bound the query time for space $O(n \cdot \polylog n)$.

\item dynamic problems, where we are looking for a trade-off between
  query and update time. (In this case, our bounds are slightly weaker
  than the originals, losing a $\lg\lg n$ factor.)
\end{itemize}

\smallskip

Our reductions also imply the following \emph{new} results:
\begin{itemize}
\item an $\Omega(\lg n / \lg\lg n)$ bound for 4-dimensional range
  reporting, given space $O(n\cdot \polylog n)$. This is quite timely,
  since a recent result \cite{nekrich07range} solved 3D reporting in
  $O(\lg^2 \lg n)$ time, raising the prospect that higher dimensions
  could also be easy.

\item a tight space lower bound for the partial match problem, for
  constant query time.

\item the first lower bound for reachability oracles.
\end{itemize}

In the process, we prove optimal randomized lower bounds for lopsided
set disjointness.
\end{abstract}



\section{Introduction}

The cell-probe model can be visualized as follows. The memory is
organized into \emph{cells} (words) of $w$ bits each. A data structure
occupies a \emph{space} of $S$ cells. The CPU receives queries and,
for dynamic data structures, updates online. The CPU starts executing
each individual operation with an empty internal state (no knowledge
about the data structure), and can proceed by reading or writing
memory cells. The running time is defined to be equal to the number of
memory probes; any computation inside the CPU is free.

The predictive power of the cell-probe model (stemming from its
machine independence and information-theoretic flavor) have long
established it as the \emph{de facto} standard for data-structure
lower bounds. The end of the 80s saw the publication of two landmark
papers in the field: Ajtai's static lower bound for predecessor search
\cite{ajtai88pred}, and the dynamic lower bounds of Fredman and Saks
\cite{fredman89cellprobe}. In the 20 years that have passed,
cell-probe complexity has developed into a mature research direction,
with a substantial bibliography: we are aware of \cite{ajtai88pred,
fredman89cellprobe, miltersen93bitprobe, miltersen94incremental,
miltersen94pred, husfeldt96sums, alstrup98marked, fredman98connect,
miltersen99asymmetric, borodin99nn, alstrup99uf, chakrabarti99ann,
alstrup01dynNN, benamram01sums, barkol02nn, beame02pred,
gal03succinct, husfeldt03sums, chakrabarti04ann, liu04ann,
jayram05match, patrascu06pred, patrascu06higher, andoni06eps2n,
patrascu06loglb, patrascu07bit, patrascu07sum2d, patrascu07randpred,
sen08roundelim}. 
\begin{figure*}
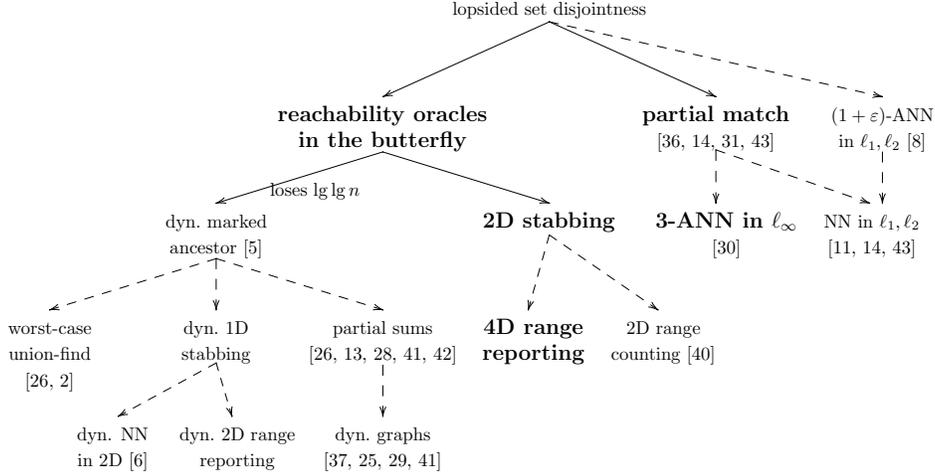

  \centering
  \scalebox{0.7}{
\Treek[5]{4}{
& & & \K{lopsided set disjointness}
     \GB{dl}{->} \GB{dr}{->} \GB{drr}{-->}
\\
& & \K{\bf\large reachability oracles}
         \Below{\bf\large in the butterfly}
  \GBk{-4.5}{0}{dl}{->}^{\textrm{\normalsize loses~$\lg\lg n$}}  
     \GBk{-4.5}{0}{dr}{->}
& &  
\K{\bf\large partial match}  \Below{\cite{miltersen99asymmetric, 
  borodin99nn,jayram05match,patrascu06higher}}
     \GBkk{0,-4.1}{0,0}{d}{-->}  \GBkk{0,-4.1}{-2.4,0}{dr}{-->}
& \K{$(1+\eps)$-ANN}\Below{in $\ell_1, \ell_2$ \cite{andoni06eps2n}} 
     \GBk{-4.5}{0}{d}{-->}
\\
& \K{dyn.~marked}\Below{ancestor \cite{alstrup98marked}}
   \GBk{-4.5}{0}{dl}{-->}    \GBk{-4.5}{0}{d}{-->}   \GBk{-4.5}{0}{dr}{-->}
& & \K{\bf\large 2D stabbing}
      \GBkk{0,0}{-4,0}{d}{-->} \GBkk{0,0}{-11,0}{dr}{-->}
& \K[2]{\bf\large 3-ANN in $\ell_\infty$}\Below{\cite{indyk01linfty}}
& \K[-2]{NN in $\ell_1, \ell_2$}
      \Below{\cite{barkol02nn,borodin99nn,patrascu06higher}} 
\\
\K{worst-case}\Below{union-find}\Below{\cite{fredman89cellprobe,alstrup99uf}}
& \K{dyn.~1D}\Below{stabbing}  
    \GBkk{0,-4.1}{3,0}{d}{-->} \GBkk{0,-4.1}{13,0}{dl}{-->}
& \K{partial sums}\Below{\cite{fredman89cellprobe, 
        benamram01sums, husfeldt03sums, patrascu06loglb, patrascu07bit}} 
    \GBk{-4.5}{0}{d}{-->}
& \K[-3]{\bf\large 4D range}\Below{\bf\large reporting}
& \K[-10]{2D range}\Below{counting \cite{patrascu07sum2d}} 
\\
\K[12]{dyn. NN}\Below{in 2D \cite{alstrup01dynNN}}
& \K[4]{dyn.~2D range}\Below{reporting}
& \K{dyn.~graphs}\Below{\cite{miltersen94incremental, 
    fredman98connect, husfeldt96sums, patrascu06loglb}}
}
  }
\caption{Dashed lines indicate reductions that were already known,
  while solid lines indicate novel reductions. For problems in bold,
  we obtain stronger results than what was previously known. }
  \label{fig:reductions}
\end{figure*}

The topics being studied cluster into three main categories:

\paragraph{Dynamic problems.}
Here, the goal is to understand the trade-off between the query time
$t_q$, and the update time $t_u$. The best known lower bound
\cite{patrascu06loglb} implies that $\max \{t_q, t_u \} = \Omega(\lg
n)$. Most proofs employ a technique introduced by Fredman and
Saks~\cite{fredman89cellprobe}, which divides the time line into
``epochs'', and argues that a query needs to read a cell written in
each epoch, lest it will miss an important update that happened then.

\paragraph{High-dimensional static problems.}
These are ``hard problems,'' exhibiting a sharp phase transition:
either the space is linear and the query is very slow (e.g.~linear
search); or the space is very large (superpolynomial) and the query is
essentially constant. 

Proofs employ a technique introduced by
Miltersen~\cite{miltersen94pred}, which considers a communication game
between a party holding the query, and a party holding the database.
Simulating the CPU's cell probes, the querier can solve the problem by
sending $t_q \lg S$ bits. If we lower bound this communication by some
$a$, we conclude that $S \ge 2^{\Omega(a/t_q)}$. The bounds are
interesting (often tight) for constant query time, but degrade quickly
for higher $t_q$.

\paragraph{Low-dimensional static problems.}
These are problems for which we have polylogarithmic query bounds,
with near linear space. The main research goal (within reach) has been
to find the best query time for space $S = O(n\cdot
\polylog n)$. The best known bound \cite{patrascu06higher} implies
that $t_q = \Omega(\lg n / \lg\lg n)$. The technique used in this
regime, introduced by P\v{a}tra\c{s}cu and
Thorup~\cite{patrascu06pred}, is to consider a direct sum
communication game, in which $O(n/\polylog n)$ queriers want to
communicate with the database simultaneously.

\bigskip

The cross-over of techniques between the three categories has so far
been minimal. At the same time, the diversity inside each category
appears substantial, even to someone well-versed in the
field. However, we will see that this appearance is deceiving.

By a series of well-crafted reductions (Figure~\ref{fig:reductions}),
we are able to unify a large majority of the known results in each of
the three categories. Since the problems mentioned in
Figure~\ref{fig:reductions} are rather well known, we do not describe
them here. The reader unfamiliar with the field can consult
Appendix~\ref{app:catalog}, which introduces these problems and
sketches some of the known reductions.

All our results follow, by reductions, from a \emph{single} lower
bound on the communication complexity of lopsided (asymmetric) set
disjointness. In this problem, Alice and Bob receive two sets $S$,
respectively $T$, and they want to determine whether $S\cap T =
\emptyset$. The lopsided nature of the problem comes from the set
sizes, $|S| \ll |T|$, and from the fact that Alice may communicate
much less than Bob.

For several problems, our unified proof is in fact simpler than the
original. This is certainly true for 2D range counting
\cite{patrascu07sum2d}, and arguably so for exact nearest neighbor
\cite{barkol02nn} and marked ancestor \cite{alstrup98marked} (though
for the latter, our bound is suboptimal by a $\lg\lg n$ factor).

For 2D stabbing and 4D range reporting, we obtain the first nontrivial
lower bounds, while for partial match, we improve the best previous
bound \cite{jayram05match} to an optimal one.

It seems safe to say that the sweeping generality of our results come
as a significant surprise (it has certainly been a major source of
surprise for the author). \emph{A priori}, it seems hard to imagine a
formal connection between such lower bounds for very different
problems, in very different settings. Much of the magic of our results
lies in defining the right link between the problems: reachability
queries in butterfly graphs. Once we decide to use this middle ground,
it is not hard to give reductions to and from set disjointness,
dynamic marked ancestor, and static 4-dimensional range
reporting. Each of these reductions is natural, but the combination is
no less surprising.

\subsection{New Results}

\paragraph{Partial match.}
Remember that in the partial match problem, we have a data base of $n$
strings in $\{0,1\}^d$, and a query string from the alphabet
$\{0,1,\star\}^d$. The goal is to determine whether any string in the
database matches this pattern, where $\star$ can match anything.  In
\S\ref{sec:pm}, we show that:

\begin{theorem}
Let Alice hold a string in $\{0,1,\star\}^d$, and Bob hold $n$ points
in $\{0,1\}^d$. In any bounded-error protocol answering the partial
match problem, either Alice sends $\Omega(d)$ bits or Bob sends
$\Omega(n^{1-\eps})$ bits, for any constant $\eps > 0$.
\end{theorem}

By the standard relation between asymmetric communication complexity
and cell-probe data structures~\cite{miltersen99asymmetric} and
decision trees~\cite{patrascu08Linfty}, this bounds implies that:
\begin{itemize}
\item a data structure for the partial match problem with cell-probe
  complexity $t$ must use space $2^{\Omega(d/t)}$, assuming the word
  size is $O(n^{1-\eps} / t)$.

\item a decision tree for the partial match problem must have size
  $2^{\Omega(d)}$, assuming the depth is $O(n^{1-2\eps}/d)$ and the
  predicate size is $O(n^\eps)$.
\end{itemize}

As usual with such bounds, the cell-probe result is optimal for
constant query time, but degrades quickly with $t$. Note that in the
decision tree model, we have a sharp transition between depth and
size: when the depth is $O(n)$, linear size can be achieved (search
the entire database).

The partial match problem is well investigated
\cite{miltersen99asymmetric, borodin99nn, jayram05match,
patrascu06higher}. The best previous bound \cite{jayram05match} for
Alice's communication was $\Omega(d/\lg n)$ bits, instead of our
optimal $\Omega(d)$. 

Our reduction is a simple exercise, and it seems surprising that the
connection was not established before. For instance, Barkol and
Rabani~\cite{barkol02nn} gave a difficult lower bound for exact near
neighbor in the Hamming cube, though it was well known that partial
match reduces to exact near neighbor. This suggests that partial match
was viewed as a ``nasty'' problem.

By the reduction of \cite{indyk01linfty}, lower bounds for partial
match also carry over to near neighbor in $\ell_\infty$, with
approximation $\le 3$. See \cite{patrascu08Linfty} for the case of
higher approximation.

\paragraph{Reachability oracles.}
The following problem appears very hard: preprocess a sparse directed
graph in less than $n^2$ space, such that reachability queries (can
$u$ be reached from $v$?) are answered efficiently. The problem seems
to belong to folklore, and we are not aware of any nontrivial positive
results. By contrast, for undirected graphs, many oracles are known.

In \S\ref{sec:butterfly}, we show the first lower bound supporting the
apparent difficulty of the problem:

\begin{theorem}   \label{thm:oracles}
A reachability oracle using space $S$ in the cell probe model with
$w$-bit cells, requires query time $\Omega(\lg n / \lg \frac{Sw}{n})$.
\end{theorem}

The bound holds even if the graph is a subgraph of a butterfly graph,
and in fact it is tight for this special case. If constant time is
desired, our bounds shows that the space needs to be $n^{1 +
\Omega(1)}$. This stands in contrast to undirected graphs, for which
connectivity oracles are easy to implement with $O(n)$ space and
$O(1)$ query time. Note however, that our lower bound is still very
far from the conjectured hardness of the problem.

\paragraph{Range reporting in 4D.}
Range reporting in 2D can be solved in $O(\lg\lg n)$ time and almost
linear space~\cite{alstrup00range}; see \cite{patrascu07randpred} for
a lower bound on the query time. 

Known techniques based on range trees can raise~\cite{alstrup00range}
a $d$-dimensional solution to a solution in $d+1$ dimensions, paying a
factor $O(\lg n / \lg\lg n)$ in time and space. It is generally
believed that this cost for each additional the dimension is optimal.
Unfortunately, we cannot prove optimal lower bounds for large $d$,
since current lower bound techniques cannot show bounds exceeding
$\Omega(\lg n / \lg\lg n)$. Then, it remains to ask about optimal
bounds for small dimension.

Until recently, it seemed safe to conjecture that 3D range reporting
would require $\Omega(\lg n / \lg\lg n)$ query time for space
$O(n\cdot \polylog n)$. Indeed, a common way to design a
$d$-dimensions static data structure is to perform a plane sweep on
one coordinate, and maintain a dynamic data structure for $d-1$
dimensions. The data structure is then made persistent, transforming
update time into space. But it was known, via the marked ancestor
problem~\cite{alstrup98marked}, that dynamic 2D range reporting
requires $\Omega(\lg n / \lg\lg n)$ query time. Thus, static 3D
reporting was expected to require a similar query time.

However, this conjecture was refuted by a recent result of
Nekrich~\cite{nekrich07range} from SoCG'07. It was shown that 3D range
reporting can be done in doubly-logarithmic query time, specifically
$t_q = O(\lg^2 \lg n)$. Without threatening the belief that
\emph{ultimately} the bounds should grow by $\Theta(\lg n / \lg\lg n)$
per dimension, this positive result raised the intriguing question
whether further dimensions might also collapse to nearly constant time
before this exponential growth begins. 

Why would 4 dimensions be hard, if 3 dimensions turned out to be easy?
The question has a simple, but fascinating answer: butterfly graphs.
By reduction from reachability on butterfly graphs, we show in
\S\ref{sec:structure} that the gap between 3 and 4 dimensions
must be $\Omegat(\lg n)$:

\begin{theorem}  \label{thm:rep4d}
A data structure for range reporting in 4 dimensions using space $S$
in the cell probe model with $w$-bit cells, requires query time
$\Omega(\lg n / \lg \frac{Sw}{n})$.
\end{theorem}

For the main case $w = O(\lg n)$ and $S = n\cdot \polylog n$, the
query time must be $\Omega(\lg n / \lg\lg n)$. This is almost tight,
since the result of Nekrich implies an upper bound of $O(\lg n \lg\lg
n)$.

\paragraph{Range stabbing in 2D.}
In fact, our reduction from reachability oracles to 4D range reporting
goes through 2D range stabbing, for which we obtain the same bounds as
in Theorem~\ref{thm:rep4d}. There exists a simple reduction from 2D
stabbing to 2D range reporting, and thus, we recover our lower bounds
for range reporting~\cite{patrascu07sum2d}, with a much simpler proof.

\subsection{Lower Bounds for Set Disjointness}

In the set disjointness problem, Alice and Bob receive sets $S$ and
$T$, and must determine whether $S \cap T = \emptyset$. We
parameterize \emph{lopsided set disjointness} (LSD) by the size of
Alice's set $|S| = N$, and $B$, the fraction between the universe and
$N$. In other words, $S, T \subseteq [N\cdot B]$. We do not impose an
upper bound on the size of $T$, i.e.~$|T| \le N\cdot B$.

Symmetric set disjointness is a central problem in communication
complexity. While a deterministic lower bound is easy to prove, the
optimal randomized lower bound was shown in the celebrated papers of
Razborov~\cite{razborov92disj} and Kalyanasundaram and
Schnitger~\cite{kalyana92disj}, dating to 1992.  Bar-Yossef et
al.~\cite{baryossef04disj} gave a more intuitive information-theoretic
view of the technique behind these proofs.

In their seminal paper on asymmetric communication complexity,
Miltersen et al.~\cite{miltersen99asymmetric} proved an (easy)
deterministic lower bound for LSD, and left the randomized lower bound
as an ``interesting'' open problem. 

In FOCS'06, we showed~\cite{andoni06eps2n} a randomized LSD lower
bound for the case when $B \ge \poly(N)$. For such large universes, it
suffices to consider an independent distribution for Alice's and Bob's
inputs, simplifying the proof considerably.

In this paper, we show how to extend the techniques for symmetric set
disjointness to the asymmetric case, and obtain an optimal randomized
bound in all cases:

\begin{theorem}   \label{thm:rand-lsd}
Fix $\delta > 0$. In a bounded error protocol for LSD, either Alice
sends at least $\delta N\lg B$ bits, or Bob sends at least $N
B^{1-O(\delta)}$ bits.
\end{theorem}

The proof appears in \S\ref{sec:lsd-proof}, and is fairly
technical. If one is only interested in deterministic lower bounds,
the proof of Miltersen et al.~\cite{miltersen99asymmetric} suffices;
this proof is a one-paragraph counting argument. If one wants
randomized lower bounds for partial match and near-neighbor problems,
it suffices to use the simpler proof of \cite{andoni06eps2n}, since
those reductions work well with a large universe. Randomized lower
bounds for reachability oracles and the entire left subtree of
Figure~\ref{fig:reductions} require small universes ($B \ll N$), and
thus need the entire generality of Theorem~\ref{thm:rand-lsd}.

\paragraph{Organization.}
The reader unfamiliar with our problems is first referred to
Appendix~\ref{app:catalog}, which defines all problems, and summarizes
the known reductions (the dashed lines in Figure~\ref{fig:reductions}).

The remainder of this paper is organized as a bottom-up, level
traversal of the tree in Figure~\ref{fig:reductions}. (We find that
this ordering builds the most intuition for the results.)

In \S\ref{sec:structure}, we explain why butterfly graphs capture the
structure hidden in many problems, and show reductions to dynamic
marked ancestor, and static 2D stabbing.

In \S\ref{sec:special-lsd}, we consider some special cases of the LSD
problem, which are shown to be as hard as the general case, but are
easier to use in reductions. Subsequently, \S\ref{sec:pm} and
\S\ref{sec:butterfly} reduce set disjointness to partial match,
respectively reachability oracles.

Finally, \S\ref{sec:lsd-proof} gives the proof of our optimal LSD
lower bound.

\begin{figure}
  \input{tikz-butterfly.tex}
  \caption{ \mbox{A butterfly with degree $2$ and depth $4$.} }
  \label{fig:butterfly}
\vspace{-3ex}
\end{figure}

\section{The Butterfly Effect}    \label{sec:structure}

The butterfly is a well-known graph structure with high ``shuffle
abilities.''  The graph (Figure~\ref{fig:butterfly}) is specified by
two parameters: the degree $b$, and the depth $d$. The graph has $d+1$
layers, each having $b^d$ vertices. The vertices on level $0$ are
sources, while the ones on level $d$ are sinks. Each vertex except the
sinks has out-degree $d$, and each vertex except the sources has
in-degree $d$. If we view vertices on each level as vectors in
$[b]^d$, the edges going out of a vertex on level $i$ go to vectors
that may differ only on the $i$th coordinate. This ensures that there
is a unique path between any source and any sink: the path ``morphs''
the source vector into the sink node by changing one coordinate at
each level.

For convenience, we will slightly abuse terminology and talk about
``reachability oracles for $G$,'' where $G$ is a butterfly graph. This
problem is defined as follows: preprocess a subgraph of $G$, to answer
queries of the form, ``is sink $v$ reachable from source $u$?'' The
query can be restated as, ``is any edge on the unique source--sink path
missing from the subgraph?''

\subsection{Reachability Oracles to Stabbing}

The reduction from reachability oracles to stabbing is very easy to
explain formally, and we proceed to do that now. However, there is a
deeper meaning to this reduction, which will be explored in
\S\ref{sec:persist}.

\begin{reduction}
Let $G$ be a butterfly with $M$ edges. The reachability oracle problem
on $G$ reduces to 2-dimensional stabbing over $M$ rectangles.
\end{reduction}

\begin{proof}
If some edge of $G$ does not appear in the subgraph, what source-sink
paths does this cut off? Say the edge is on level $i$, and is between
vertices $(\cdots, v_{i-1}, v_i, v_{i+1}, \cdots)$ and $(\cdots,
v_{i-1}, v'_i, v_{i+1}, \cdots)$. The sources that can reach this edge
are precisely $(\star, \cdots, \star, v_i, v_{i+1}, \cdots)$, where
$\star$ indicates an arbitrary value. The sinks that can be reached
from the edge are $(\cdots, v_{i-1}, v'_i, \star, \cdots)$. The
source--sink pairs that route through the missing edge are the
Cartesian product of these two sets.

This Cartesian product has precisely the format of a 2D rectangle.  If
we read a source vector $(v_1, \dots, v_d)$ as a number in base $b$
with the most significant digit being $v_d$, the set of sources that
can reach the edge is an interval of length $b^{i-1}$. Similarly, a
sink is treated as a number with the most significant digit $v_1$,
giving an interval of length $b^{d-i}$. 

For every missing edge, we define a rectangle with the source and sink
pairs that route through it. Then, a sink is reachable from a source
iff no rectangle is stabbed by the (sink, source) point.
\end{proof}

Observe that the rectangles we constructed overlap in complicated
ways. This is in fact needed, because 2-dimensional range stabbing
with non-overlapping rectangles can be solved with query time
$O(\lg^2\lg n)$ \cite{deberg95ptloc}.

As explained in Appendix~\ref{app:catalog}, 2D range stabbing reduces
to 2D range counting and 4D range reporting.

\subsection{The Structure of Dynamic Problems}   \label{sec:persist}

The more interesting reduction is to the marked ancestor problem.  The
goal is to convert a solution to the dynamic problem into a solution
to \emph{some} static problem for which we can prove a lower bound. 

A natural candidate would be to define the static problem to be the
persistent version of the dynamic problem. Abstractly, this is defined
as follows:
\begin{description}
\item[input:] an (offline) sequence of updates to a dynamic problem,
  denoted by $u_1, \dots, u_m$.
\item[query:] a query $q$ to dynamic problem and a time stamp $\tau
  \le m$. The answer should be the answer to $q$ if it were executed
  by the dynamic data structure after updates $u_1, \dots, u_\tau$.
\end{description}

An algorithm result for making data structures persistent can be used
to imply a lower bound for the dynamic problem, based on a lower bound
for the static problem. The following is a standard persistence
result:

\begin{lemma}  \label{lem:persist}
If a dynamic problem can be solved with update time $t_u$ and query
time $t_q$, its (static) persistent version will have a solution with
space $O(m\cdot t_u)$ and query time $O(t_q \cdot \lg\lg (m\cdot
t_u))$.
\end{lemma}

\begin{proof}
We simulate the updates in order, and record their cell writes. Each
cell in the simulation has a collection of values and timestamps
(which indicate when the value was updated). For each cell, we build a
van Emde Boas predecessor structure~\cite{vEB77pred} over
the time-stamps. The structures occupy $O(m\cdot t_u)$ space in total,
supporting queries in $O(\lg\lg (mt_u))$ time. To simulate the query,
we run a predecessor query for every cell read, finding the last update
that changed the cell before time $\tau$.
\end{proof}

Thus, if the static problem is hard, so is the dynamic problem (to
within a doubly logarithmic factor). However, the reverse is not
necessarily true, and the persistent version of marked ancestor turns
out to be easy, at least for the incremental case. To see that,
compute for each node the minimum time when it becomes marked. Then,
we can propagate down to every leaf the minimum time seen on the
root-to-leaf path. To query the persistent version, it suffices to
compare the time stamp with this value stored at the leaf.

\smallskip

As it turns out, persistence is still the correct intuition for
generating a hard static problem. However, we need the stronger notion
of full persistence. In partial persistence, as seen above, the
updates create a linear chain of versions (an update always
affects the more recent version). In full persistence, the updates
create a \emph{tree} of versions, since updates are allowed to modify
any historic version.

For an abstract dynamic problem, its fully-persistent version is
defined as follows:
\begin{description}
\item[input:] a rooted tree (called the \emph{version tree}) in which
  every node is labeled with a sequence of update operations. The total
  number of updates is $m$.

\item[query:] a query $q$ to the dynamic problem, and a node $\tau$ of
  the version tree. The answer should be the answer to $q$ if it were
  executed after the sequence of updates found on the path through the
  version tree from the root to $\tau$.
\end{description}

Like the partially persistent problem, the fully persistent one can be
solved by efficient simulation of the dynamic problem:

\begin{lemma}  \label{lem:full-persist}
If a dynamic problem can be solved with update time $t_u$ and query
time $t_q$, the fully-persistent static problem has a solution with
space $O(m\cdot t_u)$ and query time $O(t_q \lg\lg (m\cdot t_u))$.
\end{lemma}

\begin{proof}
For each cell of the simulated machine, consider the various nodes of
the version tree in which the cell is written. Given a ``time stamp''
(node) $\tau$, we must determine the most recent change that happened
on the path from $\tau$ to the root. This is the longest matching
prefix problem, which is equivalent to static predecessor
search. Thus, the simulation complexity is the same as in
Lemma~\ref{lem:persist}.
\end{proof}

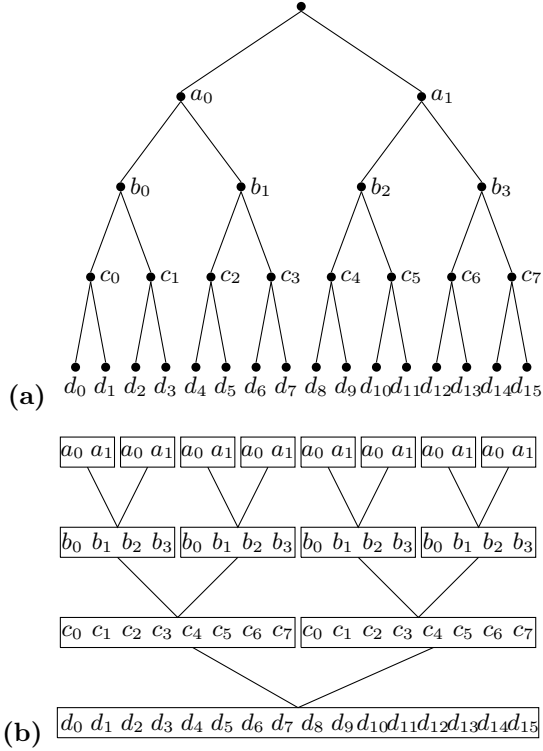
\begin{figure*}
  \input{tikz-persistent.tex}
  \caption{ (a) The marked ancestor problem.~~ (b) An instance of 
    fully-persistent marked ancestor.}
  \label{fig:marked}
\end{figure*}

We now have to prove a lower bound for the fully-persistent version of
marked ancestor, which we accomplish by a reduction from reachability
oracles in the butterfly:

\begin{reduction}   \label{red:full-persist}
Let $G$ be a subgraph of a butterfly with $M$ edges. The reachability
oracle problem on $G$ reduces to the fully-persistent version of the
marked ancestor problem, with an input of $O(M)$ offline updates. The
tree in the marked ancestor problem has the same degree and depth as
the butterfly.
\end{reduction}

\begin{proof}
Our inputs to the fully-persistent problem have the pattern
illustrated in Figure~\ref{fig:marked}. At the root of the version
tree, we have update operations for the leaves of the marked ancestor
tree. If we desire a lower bound for the incremental marked ancestor
problems, all nodes start unmarked, and we have an update for every
leaf that needs to be marked. If we want a decremental lower bound,
all nodes start marked, and all operations are \emph{unmark}.

The root has $b$ subversions; in each subversion, the level above the
leaves in the marked ancestor tree is updated. The construction
continues similarly, branching our more versions at the rate at which
level size decreases. Thus, on each level of the version tree we have
$b^d$ updates, giving $b^d \cdot d$ updates in total.

With this construction of the updates, the structure of the fully
persistent marked ancestor problem is isomorphic to a
butterfly. Imagine what happens when we query a leaf $v$ of the marked
ancestor tree, at a leaf $t$ of the version tree. We think of both $v$
and $t$ as vectors in $[b]^d$, spelling out the root to leaf paths.
The path from the root to $v$ goes through every level of
the version tree:
\begin{itemize}
\item on the top level, there is a single version ($t$ is irrelevant),
  in which $v$ is updated.

\item on the next level, the subversion we descend to is decided by
  the first coordinate of $t$. In this subversion, $v$'s parent is
  updated. Note that $v$'s parent is determined by the first $d-1$
  coordinates of $v$.

\item on the next level, the relevant subversion is dictated by the
  first two coordinates of $t$. In this subversion, $v$'s grandparent
  is updated, which depends on the first $d-2$ coordinates of $v$.

\item etc.
\end{itemize}

\noindent
This is precisely the definition of a source-to-sink path in the
butterfly graph, morphing the source into the sink one coordinate at a
time. Each update will mark a node if the corresponding edge in the
butterfly is missing in the subgraph. Thus, we encounter a marked
ancestor iff some edge is missing.
\end{proof}

Let us see how Reduction~\ref{red:full-persist} combines with
Lemma~\ref{lem:full-persist} to give a lower bound for the dynamic
marked ancestor problem. Given a butterfly graph with $m$ edges, we
generate at most $m$ updates. From Lemma~\ref{lem:full-persist}, the
space of the fully persistent structure is $S = O(m\cdot t_u)$, and
the query time $O(t_q \lg\lg (mt_q))$, where $t_u$ and $t_q$ are the
assumed running times for dynamic marked ancestor. If $t_u = \polylog
m$, the space is $S = O(m\polylog m)$.

The lower bound for reachability oracles from
Theorem~\ref{thm:oracles} implies that for space $O(m\polylog m)$, the
query time must be $\Omega\big( \frac{\lg m}{\lg\lg m} \big)$. But we
have an upper bound of $O(t_q \lg\lg (mt_q))$ for the query time, so
$t_q = \Omega\big( \frac{\lg m}{\lg^2 \lg m} \big)$. This is weaker by
a $\lg\lg m$ factor compared to the original bound of
\cite{alstrup98marked}.

\section{Adding Structure to Set Disjointness}   \label{sec:special-lsd}

Just as it is more convenient to work with Planar-NAE-3SAT that
Circuit-SAT for showing NP-completeness, our reductions use two
restricted versions of LSD:
\begin{description}
\item[Blocked-LSD:] The universe is interpreted as $[N] \times [B]$,
  and elements as pairs $(u,v)$. It is guaranteed that $(\forall) x
  \in [N]$, $S$ contains a single element of the form $(x, \star)$.

\item[2-Blocked-LSD:] The universe is interpreted as $[\frac{N}{B}]
  \times [B] \times [B]$. It is guaranteed that for all $x \in
  [\frac{N}{B}]$ and $y\in [B]$, $S$ contains a single element of the
  form $(x, y, \star)$ and a single element of the form $(x, \star,
  y)$.
\end{description}

It is possible to reanalyze the lower bound of \S\ref{sec:lsd-proof}
and show directly that it applies to these restricted
versions. However, in the spirit of the paper, we choose to design a
reduction from general LSD to these special cases.

\begin{lemma}
LSD reduces to Blocked-LSD by a deterministic protocol with
communication complexity $O(N)$.
\end{lemma}

\begin{proof}
In the general LSD, Alice's set $S$ might contain multiple elements in
each block. Alice begins by communicating to Bob the vector $(c_1,
\dots, c_N)$, where $c_i$ denotes the number of elements in block
$i$. The number of distinct possibilities for $(c_1, \dots, c_N)$ is
$\binom{2N-1}{N}$, so Alice needs to send $O(N)$ bits (in a possibly
non-uniform protocol).

Now Bob constructs a set $T'$ in which the $i$-th block of $T$ is
included $c_i$ times; a block with $c_i = 0$ is discarded.  Alice
considers a set $S'$ in which block $i$ gets expanded into $c_i$
blocks, with one element from the original block appearing in each of
the new blocks. We now have an instance of Blocked-LSD.
\end{proof}

\begin{lemma}
Blocked-LSD reduces to 2-Blocked-LSD by a deterministic protocol with
communication complexity $O(N)$.
\end{lemma}

\begin{proof}
Consider $B$ consecutive blocks of Blocked-LSD. Adjoining these blocks
together, we can view the universe as a $B \times B$ matrix. The
matrix has one entry in each column (one entry per block), but may
have multiple entries per row. The protocol from above can be applied
to create multiple copies of rows with more elements. After the
protocol is employed, there is one element in each row and each
column. Doing this for every group of $B$ blocks, the total
communication will be $\frac{N}{B} \cdot O(B) = O(N)$.
\end{proof}

Since the lower bound for LSD says that Alice must communicate
$\omega(N)$ bits, these reductions show that Blocked-LSD and
2-Blocked-LSD have the same complexity.

\subsection{Reductions}

Before proceeding, we must clarify the notion of reduction from a
communication problem to a data-structure problem. In such a
reduction, Bob constructs a database based on his set $T$, and Alice
constructs a set of $k$ queries. It is then shown that LSD can be
solved based on the answer to the $k$ queries on Bob's database.

When analyzing data structures of polynomial space or more, we will in
fact use just one query ($k=1$). If the data structure has size $S$
and query time $t$, this reduction in fact gives a communication
protocol for LSD, in which Alice communicates $t \lg S$ bits, and Bob
communicates $tw$ bits. This is done by simulating the query
algorithm: for each cell probe, Alice sends the address, and Bob sends
the content from his constructed database. At the end, the answer to
LSD is determined from the answer of the query.

If we are interested in lower bounds for space $n^{1+o(1)}$, note that
an upper bound of $\lg S$ for Alice's communication no longer
suffices, because $S = O(n^{1+\eps})$ and $S = O(n)$ yield the same
asymptotic bound. The work-around is to reduce to $k$ \emph{parallel}
queries, for large $k$. In each cell probe, the queries want to read
some $k$ cells from the memory of size $S$. Then, Alice can send $\lg
\binom{S}{k}$ bits, and Bob can reply with $k\cdot w$. Observe that
$\lg \binom{S}{k} \ll k\lg S$, if $k$ is large enough.

\section{Set Disjointness to Partial Match}  \label{sec:pm}
~

\begin{reduction}
Blocked-LSD reduces to one partial match query over $n=N\cdot B$
strings in dimension $d = O(N\lg B)$.
\end{reduction}

\begin{proof}
Consider a constant weight code $\phi$ mapping the universe $[B]$ to
$\{0,1\}^b$. If we use weight $b/2$, we have $\binom{b}{b/2} =
2^{\Omega(b)}$ codewords. Thus, we may set $b = O(\lg B)$.

If $S = \{ (1, s_1), \dots, (N, s_N) \}$, Alice constructs the query
string $\phi(s_1) \phi(s_2) \cdots$, i.e.~the concatenation of the
codewords of each $s_i$. We have dimension $d = N\cdot b = O(N \lg
B)$.

For each point $(x,y) \in T$, Bob places the string $0^{(x-1) b} \,
\phi(y)\, 0^{(N-x) b}$ in the database. Now, if $(i, s_i) \in T$, the
database contains a string with $\phi(s_i)$ at position $(i-1) b$, and
the rest zeros. This string is dominated by the query, which also has
$\phi(s_i)$ at that position. On the other hand, if a query dominates
some string in the database, then for some $(i, s_i) \in S$ and $(i,
y) \in T$, $\phi(s_i)$ dominates $\phi(y)$. But this means $s_i = y$
because in a constant weight code, no codeword can dominate another.
\end{proof}

From the lower bound on Blocked-LSD, we know that in a communication
protocol solving the problem, either Alice sends $\Omega(N \lg B)$
bits, or Bob sends $N \cdot B^{1-\delta} \ge n^{1-\delta}$
bits. Rewriting this bound in terms of $n$ and $d$, either Alice sends
$\Omega(d)$ bits, or Bob sends $n^{1-\delta}$ bits, for constant
$\delta > 0$.

This implies that a data structure with query time $t$ requires space
$2^{\Omega(d/t)}$, as long as the word size is $w \le n^{1-\delta} /
t$. It also implies that any decision tree of depth $n^{1-\delta}$
needs to have size $2^{\Omega(d/t)}$.

\section{Set Disjointness to Reachability Oracles} \label{sec:butterfly}

Since we want a lower bound for near-linear space, we must reduce LSD
to $k$ parallel queries on the reachability oracle. The entire action
is in what value of $k$ we can achieve. Note, for instance, that $k =
N$ is trivial, because Alice can pose a query for each item in her
set. However, a reduction with $k=N$ is also useless. Remember that
the communication complexity of Alice is $t \cdot \lg \binom{S}{k} \ge
t \lg \binom{NB}{N}$. But LSD is trivially solvable with communication
$\lg \binom{NB}{N}$, since Alice can communicate her entire set. Thus,
there is no contradiction with the lower bound.

To get a lower bound on $t$, $k$ must be made as small as possible
compared to $N$. Intuitively, a source--sink path in a butterfly of
depth $d$ traverses $d$ edges, so it should be possible to test $d$
elements by a single query. To do that, the edges must assemble in
contiguous source--sink paths, which turns out to be possible if we
carefully match the structure of the butterfly and the 2-Blocked-LSD
problem:

\begin{reduction}
Let $G$ be a degree-$B$ butterfly graph with $N$ non-sink vertices and
$N\cdot B$ edges, and let $d$ be its depth. 2-Blocked-LSD reduces to
$\frac{N}{d}$ parallel queries to a reachability oracle for a subgraph
of $G$.
\end{reduction}

\begin{proof}
Remember that in 2-Blocked-LSD, elements are triples $(x,y,z)$ from
the universe $[\frac{N}{B}] \times [B] \times [B]$. We define below a
bijection between $[\frac{N}{B}] \times [B]$ and the non-sink vertices
of $G$. Since $(x,y)$ is mapped to a non-sink vertex, it is natural to
associate $(x,y,z)$ to an edge, specifically edge number $z$ going out
of vertex $(x,y)$.

Bob constructs a reachability oracle for the graph $G$ excluding the
edges in his set $T$. Then, Alice must find out whether any edge from
her set $S$ has been deleted. By mapping the universe $[\frac{N}{B}]
\times [B]$ to the nodes carefully, we will ensure that Alice's edges
on each level form a perfect matching. Then, her set of $N$ edges form
$\frac{N}{d}$ disjoint paths from sources to sinks. Using this
property, Alice can just issue $\frac{N}{d}$ queries for these
paths. If any of the source--sink pairs is unreachable, some edge in
$S$ has been deleted.

To ensure Alice's edges form perfect matchings at each level, we first
decompose the non-sink vertices of $G$ into $\frac{N}{B}$ microsets of
$B$ elements each. Each microset is associated to some level $i$, and
contains nodes of the form $(\cdots, v_{i-1}, \star, v_{i+1}, \cdot)$
on level $i$. A value $(x,y)$ is mapped to node number $y$ in a
microset identified by $x$ (through some arbitrary bijection between
$[\frac{N}{B}]$ and microsets).

Let $(x, 1, z_1), \dots, (x, B, z_B)$ be the values in $S$ that give
edges going out of microset $x$. If the nodes of the microset are the
vectors $(\cdots, v_{i-1}, \star, v_{i+1}, \cdot)$, the nodes to which
the edges of $S$ go are the vectors $(\cdots, v_{i-1}, z_j, v_{i+1},
\cdot)$ on the next level, where $j \in [B]$. Observe that edges from
different microsets cannot go to the same vertex. Also, edges from the
same microset go to distinct vertices by the 2-Blocked property: for
any fixed $x$, the $z_j$'s are distinct. Since all edges on a level
point to distinct vertices, they form a perfect matching.
\end{proof}

Let us now compute the lower bounds implied by the reduction. We
obtain a protocol for 2-Blocked-LSD in which Alice communicates $t \lg
\binom{S}{k} = O(tk \lg\frac{S}{k}) = O(N \cdot \frac{t}{d} \lg
\frac{Sd}{N})$ bits, and Bob communicates $k\cdot t\cdot w = O(N \cdot
\frac{t}{d} \cdot w)$ bits. On the other hand, the lower bound for
2-Blocked-LSD says that Alice needs to communicate $\Omega(N\lg B)$
bits, or Bob needs to communicate $N B^{1-\delta}$, for any constant
$\delta > 0$. It suffices to use, for instance, $\delta =
\frac{1}{2}$.

Comparing the lower bounds with the reduction upper bound, we conclude
that either $\frac{t}{d}\lg \frac{Sd}{N} = \Omega(\lg B)$, or
$\frac{t}{d} w = \Omega(\sqrt{B})$. Set the degree of the butterfly to
satisfy $B \ge w^2$ and $\lg B \ge \lg \frac{Sd}{N}$. Then,
$\frac{t}{d} = \Omega(1)$, i.e.~$t = \Omega(d)$. This is intuitive: it
shows that the query needs to be as slow as the depth, essentially
traversing a source to sink path.

Finally, note that the depth is $d = \Theta(\log_B N)$. Since $\lg B
~\ge~ \max \big\{ 2\lg w, \lg \frac{Sd}{N} \big\} ~=~ \Omega\big( \lg
w + \lg \frac{Sd}{N} \big) ~=~ \Omega\big( \lg \frac{Sdw}{N}
\big)$. Note that certainly $d < w$, so $\lg B = \Omega( \lg
\frac{Sw}{N})$. We obtain $t = \Omega(d) = \Omega(\lg N / \lg
\frac{Sw}{N})$.

\section{Proof of the LSD Lower Bounds} \label{sec:lsd-proof}

Our goal here is to prove Theorem~\ref{thm:rand-lsd}, our optimal
lower bound for LSD.

\subsection{The Hard Instances}

We imagine the universe to be partitioned into $N$ \emph{blocks}, each
containing $B$ elements. Alice's set $S$ will contain exactly one
value from each block. Bob's set $T$ will contain $\frac{B}{2}$ values
from each block; more precisely, it will contain one value from each
pair $\{ (j, 2k); (j, 2k+1) \}$.

Let $\calS$ and $\calT$ be the possible choices for $S$ and $T$
according to these rules. Note that $|\calS| = B^N$ and $|\calT| =
2^{NB/2}$. We denote by $S_i$ Alice's set restricted to block $i$, and
by $T_i$ Bob's set restricted to block $i$. Let $\calS_i$ and
$\calT_i$ be the possible choices for $S_i$ and $T_i$. We have
$|\calS_i| = B$ and $|\calT_i| = 2^{B/2}$.

We now define $\Dyes$ to be the uniform distribution on pairs $(S,
T) \in \calS \times \calT$ with $S \cap T = \emptyset$. In each block
$i$, there are two natural processes to generate $(S_i, T_i) \in
\calS_i \times \calT_i$ subject to $S_i \cap T_i = \emptyset$:

\begin{enumerate}
\item Pick $T_i \in \calT_i$ uniformly at random, i.e.~independently
  pick one element from each pair $\{(i,2k), (i,2k+1)\}$. Then, pick
  the singleton $S_i$ uniformly at random from the complement of
  $T_i$. Note that $\HH(S_i \mid T_i) = \log_2 (B/2)$.

\item Pick $S_i$ to be a uniformly random element from block
  $i$. Then, pick $T_i$ such that it doesn't intersect
  $S_i$. Specifically, if $S_i \cap \{2k, 2k+1\} = \emptyset$, $T_i$
  contains a random element among $2k$ and $2k+1$. Otherwise, $T_i$
  gets the element not in $S_i$. Note that $\HH(T_i \mid S_i) =
  \frac{B}{2} - 1$.
\end{enumerate}

To generate the distribution $\Dyes$, we will employ the following
process. First, pick $q \in \{0,1\}^N$ uniformly at random. For each
$q_i = 0$, apply process 1.~from above in block $i$; for each $q_i =
1$, apply process 2.~in block $i$. Now let $Q$ be a random variable
entailing: the vector $q$; the value $S_i$ for every $i$ with $q_i =
0$; and the value $T_i$ for every $i$ with $q_i = 1$. Intuitively, $Q$
describes the ``first half'' of each random process.

We now define distributions $\calD_k$ as follows. In block $k$ (called
\emph{the designated block}), choose $(S_k,T_k) \in \calS_k \times
\calT_k$ uniformly. Notice that $\Pr[S_k \cap T_k \ne \emptyset] =
\frac{1}{2}$. In all other blocks $i \ne k$, choose $(S_i, T_i) \in
\calS_i \times \calT_i$ as in the distribution $\Dyes$ above. As
above, we have a vector $Q_{-k}$, containing: $q_i$ for $i\ne k$; all
$S_i$ such that $q_i = 0$; and all $T_i$ such that $q_i = 1$.

We are going to prove that:

\begin{theorem}    \label{thm:lsd-distr}
Fix $\delta > 0$. If a protocol for LSD has error less than
$\frac{1}{9999}$ on distribution $\frac{1}{N} \sum_{i=1}^N \calD_i$,
then either Alice sends at least $\delta N\lg B$ bits, or Bob sends at
least $N\cdot B^{1-O(\delta)}$ bits.
\end{theorem}

The distribution $\Dyes$ will be used to measure various entropies in
the proof, which is convenient because the blocks are
independent. However, the hard distribution on which we measure error
is the mixture of $\calD_i$'s. (Since $\Dyes$ only has yes instances,
measuring error on it would be meaningless.) While it may seem
counterintutive that we argue about entropies on one distribution and
error on another, remember that $\Dyes$ and $D_i$ are not too
different: $S$ and $T$ are disjoint with probability $\frac{1}{2}$
when chosen by $\calD_i$.

\subsection{A Direct Sum Argument}

We now wish to use a direct-sum argument to obtain a low-communication
protocol for a single subproblem on $\calS_i \times
\calT_i$. Intuitively, if the LSD problem is solved by a protocol in
which Alice and Bob communicate $a$, respectively $b$ bits, we might
hope to obtain a protocol for some subproblem $i$ in which Alice
communicates $O(\frac{a}{N})$ bits and Bob communicates
$O(\frac{b}{N})$ bits.

Let $\pi$ be the transcript of the communication protocol. If Alice
sends $a$ bits and Bob $b$ bits, we claim that $\I_{\Dyes} (S : \pi
\mid Q) \le a$ and $\I_{\Dyes} (T : \pi \mid Q) \le b$. Indeed, once we
condition on $Q$, $S$ and $T$ are independent random variables: in each
block, either $S$ is fixed and $T$ is random, or vice versa. The
independence implies that all information about $S$ is given by
Alice's messages, and all information about $T$ by Bob's messages.

Define $S_{<i} = (S_1, \dots, S_{i-1})$.  We can decompose the mutual
information as follows: $\I_{\Dyes} (S: \pi \mid Q) = \sum_{i=1}^N
\I_{\Dyes} (S_i : \pi \mid Q, S_{<i})$. The analogous relation holds
for $T$. By averaging, it follows that for at least half of the values
of $i$, we simultaneously have:
\begin{equation} \label{eq:inf-both}
\I_{\Dyes} (S_i :\pi \mid Q, S_{<i} ) \le \frac{4a}{N} 
\qquad \textnormal{and} \qquad 
\I_{\Dyes} (T_i : \pi \mid Q, T_{<i} ) \le \frac{4b}{N}.
\end{equation}
Remember that the average error on $\frac{1}{N} \sum_i \calD_i$ is
$\frac{1}{9999}$. Then, there exists $k$ among the half satisfying
\eqref{eq:inf-both}, such that the error on $\calD_k$ is at most
$\frac{2}{9999}$. For the remainder of the proof, fix this $k$.

We can now reinterpret the original protocol for LSD as a new protocol
for the disjointness problem in block $k$. This protocol has the
following features:

\begin{description}
\item[Inputs:] Alice and Bob receive $S_k \in \calS_k$, respectively
  $T_k \in \calT_k$.

\item[Public coins:] The protocol employs public coins to select
  $Q_{-k}$. For every $i < k$ with $q_i = 0$, $S_i$ is chosen publicly
  to be disjoint from $T_i$ (which is part of $Q_{-k}$). For every $i
  < k$ with $q_i = 1$, $T_i$ is chosen publicly to be disjoint from
  $S_i$. 

\item[Private coins:] Alice uses private coins to select $S_i$ for all
  $i > k$ with $q_i = 0$. Bob uses private coins to select $T_i$ for
  all $i > k$ with $q_i = 1$. As above, $S_i$ is chosen to be disjoint
  from $T_i$ (which is public knowledge, as part of $Q_{-k}$), and
  analogously for $T_i$.

\item[Error:] When $S_k$ and $T_k$ are chosen independently from
  $\calS_k \times \calT_k$, the protocol computes the disjointness of
  $S_k$ and $T_k$ with error at most $\frac{2}{9999}$. Indeed, the
  independent choice of $S_k$ and $T_k$, and the public and private
  coins realize exactly the distribution $\calD_k$.

\item[Message sizes:] Unfortunately, we cannot conclude that the
  protocol has small communication complexity in the regular sense,
  i.e.~that the messages are small. We will only claim that the
  messages have small \emph{information complexity}, namely that they
  satisfy \eqref{eq:inf-both}.
\end{description}

Observe that the disjointness problem in one block is actually the
indexing problem: Alice receives a single value (as the set $S_k$) and
she wants to determined whether that value is in Bob's set. Since $|S_k|
= 1$, we note that $S_k \cap T_k = \emptyset$ iff $S_k \not\subset
T_k$.

\subsection{Understanding Information Complexity}

In normal communication lower bounds, one shows that if the protocol
communicates too few bits, it must make a lot of errors. In our case,
however, we must show that a protocol with small information
complexity (but potentially large messages) must still make a lot of
error.

Let us see what the information complexity of \eqref{eq:inf-both}
implies. We have:
\begin{eqnarray*}
\I_{\Dyes} (S_k : \pi \mid Q, S_{<i} ) &=&
    \tfrac{1}{2} \cdot \I_{\Dyes} (S_k : \pi 
           \mid q_k = 1, T_k, Q_{-k}, S_{<i} ) \\
&+& \tfrac{1}{2} \cdot \I_{\Dyes} (S_k : \pi 
           \mid q_k = 0, S_k, Q_{-k}, S_{<i} )
\end{eqnarray*}
The second term is zero, since $\HH(S_k \mid S_k) = 0$.  Thus, the old
bound $\I_{\Dyes} (S_k : \pi \mid Q, S_{<i}) \le \frac{4a}{N}$ can be
rewritten as $\I_{\Dyes} (S_k : \pi \mid q_k = 1, T_k, Q_{-k}, S_{<i})
\le \frac{8a}{N}$. We will now aim to simplify the left hand side of
this expression.

First observe that we can eliminate $q_k=1$ from the conditioning:
$\I_{\Dyes} (S_k : \pi \mid q_k = 1, T_k, Q_{-k}, S_{<i}) = \I_{\Dyes}
(S_k : \pi \mid T_k, Q_{-k}, S_{<i})$. Indeed, $\pi$ is a function of
$S$ and $T$ alone. In other words, it is a function of the public
coins $Q_{-k}$, the private coins, $S_k$, and $T_k$. But the
distribution of the inputs is the same for $q_k=1$ and $q_k=0$. In
particular, the two processes for generating $S_k$ and $T_k$ (one
selected by $q_k = 0$, the other by $q_k = 1$) yield the the same
distribution.

Now remember that $\Dyes$ is simply $\calD_k$ conditioned on $S_k \cap
T_k = \emptyset$. Thus, we can rewrite the information under the
uniform distribution for $S_k$ and $T_k$: $\I (S_k : \pi \mid Q_{-k},
T_k, S_k \not\subset T_k, S_{<k}) \le \frac{8a}{N}$. (To alleviate notation,
we drop subscripts for $\I$ and $\HH$ whenever uniform distributions
are used.) We are now measuring information under the same
distribution used to measure the error.

Analogously, it follows that $\I(T_k : \pi \mid Q_{-k}, S_k, S_k
\not\subset T_k, T_{<k}) \le \frac{8b}{N}$. We can now apply three
Markov bounds, and fix the public coins ($Q_{-k}$, $S_{<k}$, and
$T_{<k}$) such that all of the following hold:
\begin{enumerate}
\item the error of the protocol is at most $\frac{8}{9999}$;

\item $\I(S_k : \pi \mid T_k, S_k \not\subset T_k) \le \frac{32 a}{N}$;

\item $\I(T_k : \pi \mid S_k, S_k \not\subset T_k) \le \frac{32 b}{N}$.
\end{enumerate}

To express the guarantee of 1., define a random variable $\err$ which
is one if the protocol makes an error, and zero otherwise. Note that
$\err$ is a function $\err : S_k \times \coins_A \times T_k \times
\coins_B \to \{0,1\}$, where we defined $\coins_A$ as the set of
private coin outcomes for Alice and $\coins_B$ as the private coin
outcomes for Bob. By 1., we have $\E[\err] \le \frac{8}{9999}$.

We can rewrite 2.~by expanding the definition of information:
\begin{eqnarray*}
\I(S_k : \pi \mid T_k, S_k \not\subset T_k) &=&
  \HH(S_k \mid T_k, S_k \not\subset T_k) 
  - \HH(S_k \mid T_k, \pi, S_k \not\subset T_k) \\
&=& \log_2 \tfrac{B}{2} - \HH(S_k \mid T_k, \pi, S_k \not\subset T_k) 
\end{eqnarray*}
Applying a similar expansion to $T_k$, we conclude that:
\begin{eqnarray}
\log_2 \tfrac{B}{2} - \HH(S_k \mid T_k, \pi, S_k \not\subset T_k) 
    &\le& \tfrac{32 a}{N} \label{eq:H-Sk}
\\
(\tfrac{B}{2} - 1) - \HH(T_k \mid S_k, \pi, S_k \not\subset T_k)  
    &\le& \tfrac{32b}{N}  \label{eq:H-Tk}
\end{eqnarray}

Consider some transcript $\widetilde{\pi}$ of the communication
protocol. A standard observation in communication complexity is that
the set of inputs for which $\pi = \widetilde{\pi}$ is a combinatorial
rectangle in the truth table of the protocol: one side is a subset of
$S_k \times \coins_A$, and the other a subset of $T_k \times
\coins_B$. In any rectangle, the output of the protocol is fixed.

Observe that the probability that the output of the protocol is ``no''
is at most $\frac{1}{2}$ (the probability that $S_k$ and $T_k$
intersect) plus $\frac{8}{9999}$ (the probability that the protocol
makes an error). Discard all rectangles on which the output is ``no.''
Further discard all rectangles that fail to satisfy any of the
following:
\begin{eqnarray*}
\E[\err \mid \pi = \widetilde{\pi} ] &\le& \tfrac{64}{9999} 
  \\ 
\log_2 \tfrac{B}{2} - \HH(S_k \mid T_k, S_k \not\subset T_k, 
    \pi = \widetilde{\pi}) &\le& \tfrac{256 a}{N}
  \\
(\tfrac{B}{2} - 1) - \HH(T_k \mid S_k, S_k \not\subset T_k, 
    \pi = \widetilde{\pi}) &\le& \tfrac{256 b}{N}
\end{eqnarray*}
By the Markov bound, the mass of rectangles failing each one of these
tests is at most $\frac{1}{8}$. In total, at most $\frac{1}{2} +
\frac{8}{9999} + 3\cdot \frac{1}{8} < 1$ of the mass got
discarded. Thus, there exists a rectangle $\widetilde{\pi}$ with
answer ``yes'' that satisfies all three constraints. 

Let $\sigma$ be the distribution of $S_k$ conditioned on $\pi =
\widetilde{\pi}$, and $\tau$ be the distribution of $T_k$
conditioned on $\pi = \widetilde{\pi}$.  With this notation, we have:
\begin{enumerate}
\item $\E_{\sigma, \tau}[\err ] \le \tfrac{64}{9999}$, thus
  $\Pr_{\sigma, \tau} [ S_k \cap T_k \ne \emptyset] \le
  \tfrac{64}{9999}$.

\item $\HH_{\sigma, \tau}(S_k \mid T_k, S_k \not\subset T_k) 
    \ge \log_2 \frac{B}{2} - \tfrac{256 a}{N}$.

\item $\HH_{\sigma,\tau}(S_k \mid T_k, S_k \not\subset T_k)
    \ge (\frac{B}{2} - 1) - \tfrac{256 b}{N}$.
\end{enumerate}

In the next section, we shall prove that in every ``large enough''
rectangle (in the sense of entropy) the probability that $S_k$ and
$T_k$ intersect is noticeable:
\begin{lemma}   \label{lem:rectangle}
Let $\gamma > 0$. Consider probability distributions $\sigma$ on
support $\calS_k$, and $\tau$ on support $\calT_k$. The following
cannot be simultaneously true:
\begin{align}
\Pr_{\sigma\times \tau} [S_k \cap T_k \ne \emptyset] 
   &~\le~ \tfrac{1}{42}   \label{eq:enum-Pr} \\
\HH_{\sigma\times \tau} (S_k \mid T_k, S_k \not\subset T_k) 
   &~\ge~ (1 - \gamma) \log_2 B  \label{eq:enum-HS} \\
\HH_{\sigma\times \tau} (T_k \mid S_k, S_k \not\subset T_k) 
   &~\ge~ \tfrac{B}{2} - \tfrac{1}{840} \cdot B^{1-7\gamma} 
       \label{eq:enum-HT}
\end{align}
\end{lemma}

Since $\frac{64}{9999} \le \frac{1}{42}$, one of the following must
hold:
\begin{align*}
\log_2 \tfrac{B}{2} - \tfrac{256 a}{N} &~\le~ (1 - \gamma) \log_2 B
&~\Rightarrow~  a &~\ge~ \tfrac{\gamma}{257} \cdot N \log_2 B
\\
(\tfrac{B}{2} - 1) - \tfrac{256 b}{N} &~\le~  
    \tfrac{B}{2} - \tfrac{1}{840} \cdot B^{1-7\gamma} &~\Rightarrow~   
b &~\ge~ \tfrac{1}{216000} \cdot N \cdot B^{1-7\gamma}
\end{align*}
For $N$ and $B$ greater than a constant, it follows that either Alice
sends at least $\delta N \lg B$ bits, or Bob must send at least
$\frac{1}{216 000} N \cdot B^{1 - 1799\cdot \delta}$ bits.

\subsection{Analyzing a Rectangle}

The goal of this section is to show Lemma~\ref{lem:rectangle}. Let
$\mu_\sigma$ and $\mu_\tau$ be the probability density functions of
$\sigma$ and $\tau$.  We define $\calS^\star$ as the set of values of
$S_k$ that do not have unusually high probability according to
$\sigma$: $\calS^\star = \big\{ S_k \mid \mu_\sigma(S_k) \le 1 \big/
B^{1-7\gamma} \big\}$.  We first show that significant mass is left
in $\calS^\star$:

\begin{claim}
$\mu_\sigma(\calS^\star) \ge \frac{1}{5}$.
\end{claim}

\begin{proof}
Our proof will follow the following steps:
\begin{enumerate}
\item We find a column $\widehat{T}_k$ in which the function is mostly
  one (i.e.~typically $S_k \not\subset \widehat{T}_k$), \emph{and} in
  which the entropy $\HH_\sigma(S_k \mid S_k \not\subset
  \widehat{T}_k)$ is large.

\item The mass of elements outside $\calS^\star$ is bounded by the
  mass of elements outside $\calS^\star$ and disjoint from
  $\widehat{T}_k$, plus the mass of elements intersecting
  $\widehat{T}_k$. The latter is small by point 1.

\item There are \emph{few} elements outside $\calS^\star$ and disjoint
  from $\widehat{T}_k$, because they each have high probability. Thus,
  if their total mass were large, their low entropy would drag down
  the entropy of $\HH_\sigma(S_k \mid S_k \not\subset \widehat{T}_k)$,
  contradiction.
\end{enumerate}

To achieve step 1., we rewrite \eqref{eq:enum-Pr} and
\eqref{eq:enum-HS} as:
\begin{align*} 
\Pr_{\sigma\times \tau} [S_k \subset T_k] &~=~ 
  \E_\tau \left[ \Pr_\sigma[S_k \subset T_k ] \right] 
  &~\le~& \frac{1}{10} 
\\
\log_2 B - \HH_{\sigma\times \tau} (S_k \mid T_k, S_k \not\subset T_k) 
  &~=~ \E_\tau \left[ \log_2 B - \HH_\sigma (S_k \mid 
         S_k \not\subset T_k) \right]
  &~\le~& \gamma \log_2 B
\end{align*}
Applying two Markov bounds on $T_k$, we conclude that there exists some
$\widehat{T}_k$ such that:
\begin{equation}
\Pr_\sigma[S_k \subset \widehat{T}_k ] ~\le~ \tfrac{3}{10};
\qquad\qquad
\HH_\sigma (S_k \mid S_k \not\subset \widehat{T}_k) 
  ~\ge~ (1-3\gamma) \log_2 B
\label{eq:tilde-Tk}
\end{equation}
Define $\widehat{\sigma}$ to be the distribution $\sigma$ conditioned
on $S_k \not\subset \widehat{T}_k$.

With regards to step 2., we can write $\mu_\sigma(\calS^\star) \ge 1 -
\Pr_\sigma[S_k \not\in \calS^\star ~\land~ S_k \not\subset
\widehat{T}_k] - \Pr_\sigma[S_k \subset \widehat{T}_k]$. The latter
term is at most $\frac{3}{10}$. In step 3., we will upper bound the
former term by $\frac{1}{2}$, implying $\mu_\sigma(\calS^\star) \ge
\frac{1}{5}$.

For any variable $X$ and event $E$, we can decompose:
\begin{equation}
\HH(X) ~\le~ \Pr[E] \cdot \HH(X \mid E) 
  ~+~ \Pr[\lnot\, E] \cdot \HH(X \mid \lnot\, E) ~+~ H_b(\Pr[E]),
\label{eq:split-entropy}
\end{equation}
where $\HH_b(\cdot) \le 1$ is the binary entropy function. We apply
this relation to the variable $S_k$ under the distrubtion
$\widehat{\sigma}$, choosing $\calS^\star$ as our event $E$. We
obtain:
\[
\HH_{\widehat{\sigma}} (S_k) 
~\le~ \Pr_{\widehat{\sigma}} \big[ \calS^\star \big] \cdot 
      \HH_{\widehat{\sigma}} (S_k \mid S_k \in \calS^\star)
  ~+~ \Pr_{\widehat{\sigma}} \big[ \overline{\calS^\star}\: \big] \cdot
      \HH_{\widehat{\sigma}} (S_k \mid S_k \notin \calS^\star)
  ~+~ 1
\]
We have $\HH_{\widehat{\sigma}} (S_k \mid S_k \in \calS^\star) \le
\log_2 \frac{B}{2}$ since there are at most $\frac{B}{2}$ choices for
$S_k$ disjoint from $\widehat{T}_k$. On the other hand,
$\HH_{\widehat{\sigma}} (S_k \mid S_k \notin \calS^\star) \le (1 -
7\gamma) \log_2 B$. Indeed, there are at most $B^{1-7\gamma}$
distinct values outside $\calS^\star$, since each must have
probability exceeding $1 \big/ B^{1-7\gamma}$. We thus obtain:
\[
\HH_{\widehat{\sigma}} (S_k) 
~\le~ \Pr_{\widehat{\sigma}} \big[ \calS^\star \big] \cdot \log_2 \tfrac{B}{2}
  ~+~ \Pr_{\widehat{\sigma}} \big[ \overline{\calS^\star} \:\big] \cdot
      (1 - 7\gamma) \log_2 B
  ~+~ 1
\]
If we had $\Pr_{\widehat{\sigma}} [\overline{\calS^\star} \:] \ge
\frac{1}{2}$, we would have $\HH_{\widehat{\sigma}} (S_k) ~\le~ (1 -
3.5 \gamma) \log_2 B + 1 ~<~ (1 - 3\gamma) \log_2 B$ for large
enough $B$. But this would contradict \eqref{eq:tilde-Tk}, which
states that $\HH_{\widehat{\sigma}} (S_k) \ge (1-3\gamma) \log_2
B$.

Since $\widehat{\sigma}$ was the distribution $\sigma$ conditioned on
$S_k \not\subset \widehat{T}_k$, Bayes' rule tells us that
$\Pr_\sigma[S_k \not\in \calS^\star ~\land~ S_k \not\subset
\widehat{T}_k] \le \Pr_{\widehat{\sigma}} [S_k \notin \calS^\star \:]
\le \frac{1}{2}$.
\end{proof}

\bigskip

Let us now consider the function $f(T_k) = \E_\sigma [|S_k \cap
T_k|]$. By linearity of expectation, $f(T_k) = \sum_{x \in T_k}
\Pr_\sigma[x \in S_k] = \sum_{x\in T_k} \mu_\sigma(x)$, since $S_k$
has a single element. Since $|S_k \cap T_k| \in \{0,1\}$, we can
write:
\[ \Pr_{\sigma, \tau} [S_k \cap T_k \ne \emptyset]
~=~ \E_{\sigma, \tau} \big[ |S_k \cap T_k| \big]
~=~ \E_\tau \left[ \E_\sigma [ |S_k \cap T_k| ] \right]
~=~ \E_\tau \left[ f(T_k) \right]
\]

Thus, to reach a contradiction with \eqref{eq:enum-Pr}, we must lower
bound the expectation of $f(\cdot)$ over distribution $\tau$. Since we
do not have a good handle on $\tau$, we will approach this goal
indirectly: at first, we will completely ignore $\tau$, and analyze
the distribution of $f(T_k)$ when $T_k$ is chosen uniformly at random
from $\calT_k$. After this, we will use the high entropy of $\tau$, in
the sense of \eqref{eq:enum-HT}, to argue that the behavior on $\tau$
cannot be too different from the behavior on the uniform distribution.

The expectation of $f(\cdot)$ over the uniform distribution is simple
to calculate: $\E_{T_k \in \calT_k} [f(T_k)] = \sum_x \Pr_{T_k \in
\calT_k} [x \in T_k] \cdot \mu_\sigma(x) = \sum_x \frac{1}{2}
\mu_\sigma(x) = \frac{1}{2}$. In the sums, $x$ ranges over elements in
block $k$, each of which appears in $T_k$ with probability
$\frac{1}{2}$. Note that $\mu_\sigma$ is a probability density
function, so $\sum_x \mu_\sigma(x) = 1$.

Our goal now is to show that when $T_k$ is uniform in $\calT_k$, the
distribution of $f(\cdot)$ is tightly concentrated around its mean of
$\frac{1}{2}$, and, in particular, away from zero. We will employ a
Chernoff bound: we have $f(T_k) = \sum_{x\in T_k} \mu_\sigma(x)$, and
each $x \in T_k$ is chosen independently among two distinct
values. Thus, $f(T_k)$ is the sum of $B/2$ random elements of
$\mu_\sigma$, each chosen independently.

The limitation in applying the Chernoff bound is the value of $\max_s
\mu_\sigma(x)$, which bounds the variance of each sample. The set
$\calS^\star$ now comes handy, since we can restrict our attention to
elements $x$ with small $\mu_\sigma$.  Formally, consider
$f^\star(T_k) = \sum_{x \in T_k \cap \calS^\star}
\mu_\sigma(x)$. Clearly $f^\star(T_k)$ is a lower bound for $f(T_k)$.

The mean of $f^\star(\cdot)$ is $\E_{T_k \in \calT_k} [f^\star(T_k)] =
\sum_{x \in \calS^\star} \Pr_{T_k \in \calT_k} [x \in T_k] \cdot
\mu_\sigma(x) = \frac{1}{2} \mu_\sigma(\calS^\star) \ge \frac{1}{10}$.
When $T_k$ is uniform, $f^\star(T_k)$ is the sum of $B/2$ independent
random variables, each of which is bounded by $1 \big/ B^{1 -
7\gamma}$. By the Chernoff bound,
\begin{equation}
\Pr_{T_k \in \calT_k} [ f^\star(T_k) < \tfrac{1}{20} ] 
~<~ e^{-B^{1 - 7\gamma} \cdot \frac{1}{10} \cdot \frac{1}{8}}
~\le~ e^{-B^{1-7\gamma} / 80}
\label{eq:chernoff}
\end{equation}

Now we are ready to switch back to distribution $\tau$:

\begin{claim}
$\Pr_\tau [f^\star(T_k) < \frac{1}{20}] ~\le~ \frac{1}{2}$.
\end{claim}

\begin{proof}
The main steps of our proof are:
\begin{enumerate}
\item As in the analysis of $\calS^\star$, we find a row
  $\widehat{S}_k$ in which the function is mostly one (i.e.~typically
  $\widehat{S}_k \not\subset T_k$), \emph{and} in which the entropy
  $H_\tau (T_k \mid \widehat{S}_k \not\subset T_k)$ is large.

\item $\Pr_\tau [f^\star(T_k) < \frac{1}{20}]$ is bounded by $\Pr_\tau
  [ f^\star(T_k) < \frac{1}{20} ~\land~ \widehat{S}_k \not\subset
  T_k]$, plus the probability that $\widehat{S}_k \subset T_k$. The
  latter is small by point 1.

\item There are few distinct values of $T_k$ for which $f^\star(T_k) <
  \frac{1}{20}$. If these values had a large mass conditioned on
  $\widehat{S}_k \not\subset T_k$, they would drag down the entropy of
  $H_\tau (T_k \mid \widehat{S}_k \not\subset T_k)$.
\end{enumerate}

To achieve step 1., we rewrite \eqref{eq:enum-Pr} and
\eqref{eq:enum-HT} as:
\begin{align*} 
\Pr_{\sigma\times \tau} [S_k \subset T_k] &~=~ 
  \E_\tau \left[ \Pr_\sigma[S_k \subset T_k ] \right] 
  &~\le~& \tfrac{1}{10} 
\\
\tfrac{B}{2} - \HH_{\sigma\times \tau} (T_k \mid S_k, S_k \not\subset T_k) 
  &~=~ \E_\sigma \left[ \tfrac{B}{2} - \HH_\tau (T_k \mid S_k \not\subset T_k) 
\right]
  &~\le~& \tfrac{1}{840} \cdot B^{1 - 7\gamma}
\end{align*}
Applying two Markov bounds on $S_k$, we conclude that there exists some
$\widehat{S}_k$ such that:
\begin{equation}
\Pr_\tau [\widehat{S}_k \subset T_k ] ~\le~ \tfrac{3}{10};
\qquad\qquad
\HH_\tau (T_k \mid \widehat{S}_k \not\subset T_k) 
  ~\ge~ \frac{B}{2} - \tfrac{1}{280} \cdot B^{1 - 7\gamma}
\label{eq:tilde-Sk}
\end{equation}
Define $\widehat{\tau}$ to be the distribution $\tau$ conditioned on
$\widehat{S}_k \not\subset T_k$.

For step 2., we can write:
\begin{align*}
\Pr_\tau \big[ f^\star(T_k) < \tfrac{1}{20} \big] &~=~
\Pr_\tau \big[ f^\star(T_k) < \tfrac{1}{20} ~\land~ 
                       \widehat{S}_k \not\subset T_k \big] 
~+~ \Pr_\tau \big[ f^\star(T_k) < \tfrac{1}{20} ~\land~ 
                       \widehat{S}_k \subset T_k \big]
\\ &~\le~ 
\Pr_\tau \big[ f^\star(T_k) < \tfrac{1}{20} \mid 
      \widehat{S}_k \not\subset T_k \big] 
  ~+~ \Pr_\tau \big[ \widehat{S}_k \subset T_k \big] 
~\le~ \Pr_{\widehat{\tau}} \big[ f^\star(T_k) < \tfrac{1}{20} \big] 
        ~+~ \tfrac{3}{10}
\end{align*}

We now wish to conclude by proving that $\Pr_{\widehat{\tau}} [
f^\star(T_k) < \frac{1}{20} ] \le \frac{1}{5}$. We apply the relation
\eqref{eq:split-entropy} to the variable $T_k$ distributed according
to $\widehat{\tau}$, with the event $E$ chosen to be $f^\star(T_k)
< \frac{1}{20}$:
\[ H_{\widehat{\tau}} (T_k) ~\le~ 
\Pr_{\widehat{\tau}} \big[ f^\star(T_k) < \tfrac{1}{20} \big]  \cdot 
  H_{\widehat{\tau}} \big( T_k \mid f^\star(T_k) < \tfrac{1}{20} \big)
~+~ \Pr_{\widehat{\tau}} \big[ f^\star(T_k) \ge \tfrac{1}{20} \big] 
        \cdot \tfrac{B}{2} ~+~ 1
\]
By \eqref{eq:chernoff}, there are at most $2^{B/2} / e^{B^{1-7\gamma}
/ 80}$ distinct choices of $T_k$ such that $f^\star(T_k) <
\frac{1}{20}$. Thus, $H_{\widehat{\tau}} (T_k \mid f^\star(T_k) <
\frac{1}{20}) ~\le~ \frac{B}{2} - B^{1 - 7\gamma} \cdot \frac{\log_2
e}{80}$.

If $\Pr_{\widehat{\tau}} [f^\star(T_k) < \frac{1}{20} ] ~\ge~
\frac{1}{5}$, then $H_{\widehat{\tau}} (T_k) \le \frac{B}{2} - B^{1 -
7\gamma} \cdot \frac{\log_2 e}{400} + 1 ~<~ \frac{B}{2} - B^{1 -
7\gamma} / 280$ for sufficiently large $B$. But this contradicts
\eqref{eq:tilde-Sk}.
\end{proof}

We have just shown that $\Pr_{\sigma, \tau} [S_k \cap T_k \ne
\emptyset] = \E_\tau [ f(T_k) ] \ge \E_\tau [ f^\star(T_k) ] \ge
\frac{1}{20} \cdot \frac{1}{2} = \frac{1}{40}$. This contradicts
\eqref{eq:enum-Pr}. Thus, at least one of \eqref{eq:enum-Pr},
\eqref{eq:enum-HS}, and \eqref{eq:enum-HT} must be false.

This concludes the proof of Lemma~\ref{lem:rectangle} and of
Theorem~\ref{thm:lsd-distr}.

\section{Conclusion}

We have shown that many important lower bounds can be derived from a
single core problem, through a series of clean, conceptual
reductions. It is unclear what the ultimate value of this discovery
will be, but the following thoughts come to mind:

{\bf 1.} We are gaining understanding into the structure of the
problems at hand.

{\bf 2.} We simplify several known proofs. For example, we sidestep
the technical complications in the previous lower bounds for 2D range
counting \cite{patrascu07sum2d} and exact nearest neighbor
\cite{barkol02nn}.

{\bf 3.} We can now teach data-structure lower bounds to a broad
audience. Even ``simple'' lower bounds are seldom light on technical
details. By putting all the work in one bound, we can teach many
interesting results through clean reductions. (If we are satisfied
with deterministic bounds, the lower bound for set disjointness from
\cite{miltersen99asymmetric} is a one-paragraph counting argument.)

{\bf 4.} Our results hint at a certain degree of redundancy in our
work so far. In doing so, they also mark the borders of our
understanding particularly well, and challenge us to discover
surprising new paths that go far outside these borders.

\paragraph{Acknowledgements.}
The author would like to thank Yakov Nekrich and Marek Karpinsky for
useful discussions on the range reporting problem, and Alex Andoni and
T.S.~Jayram for useful discussions on the randomized LSD lower bound.

\bibliographystyle{plain}
\bibliography{../general}

\appendix

\section{Catalog of Problems}  \label{app:catalog}

\paragraph{Range Queries.}
Given a set of $n$ queries in $d$-dimensional space (say, $[n]^d$), we
can ask two classic queries: \emph{report} the points inside a range
$[a_1, b_1] \times \cdots \times [a_d, b_d]$, or simply \emph{count}
the number of points in the range. These queries lie at the heart of
database analysis, and any course on SQL is bound to start with an
example of the form: ``find employees born between 1980 and 1989,
whose salary is between \$80,000 and \$90,000.''

Note that if there are $k$ points inside the range, reporting them
necessarily takes time $\Omega(k)$. To avoid this technicality, in
this paper we only consider the decision version of reporting: is
there a point inside the range?

\paragraph{Stabbing queries.}
A dual of range queries is \emph{stabbing}: preprocess a set of $n$
boxes of the form $[a_1, b_1] \times \cdots \times [a_d, b_d]$, such
that we can quickly find the box(es) containing a query point. 

Stabbing is a very important form of classification queries. For
instance, network routers have rules applying to packets coming from
some IP range, and heading to another IP range. A query is needed for
every packet passing through the router, making this a critical
problem. This application has motivated several theoretically-minded
papers \cite{thorup03stab, feldmann00ip, applegate07acl,
eppstein01stabbing}, as well as a significant body of
practically-minded ones.

Another important application of stabbing is method dispatching, in
experimental object oriented languages that (unlike, say, Java and
C++) allow dynamic dispatching on more arguments than the class. This
application has motivated several theoretically-minded papers
\cite{muthu96stabbing, alstrup02stabbing, ferragina96stabbing,
ferragina99stabbing}, as well as a number of practically-minded ones.

Our lower bounds for 2D stabbing are the first for this problem, and
in fact, match the upper bound of~\cite{chazelle88functional}.

It is easy to see that stabbing in $d$ dimensions reduces to range
reporting in $2d$ dimensions, since boxes can be expressed as
$2d$-dimensional points.

The decision version of stabbing in 2D also reduces to (weighted)
range counting in 2D by the following neat trick. We replace a
rectangle $[a_1, b_1] \times [a_2, b_2]$ by 4 points: $(a_1, b_1)$ and
$(a_2, b_2)$ with weight $+1$, and $(a_1, b_2)$ and $(a_2, b_1)$ with
weight $-1$. To test whether $(q_1, q_2)$ stabs a rectangle, query the
sum in the range $[0,q_1] \times [0, q_2]$. If the query lies inside a
rectangle, the lower-left corner contributes $+1$ to count. If the
query point is outside, the corners cancel out. 

With a bit of care, the reduction can be made to work for unweighted
range counting, by ensuring the query never stabs more than one
rectangle. Then, it suffices to count points mod $2$.

\paragraph{Partial match.}
The problem is to preprocess a data base of $n$ strings in
$\{0,1\}^d$. Then, a query string from the alphabet $\{0,1,\star\}^d$
is given, and we must determine whether any string in the database
matches this pattern (where $\star$ can match anything). This is
equivalent to a problem in which the query is in $\{0,1\}^d$, and we
must test whether any string in the database is dominated by the query
(where $a$ dominates $b$ if on every coordinate $a_i \ge b_i$).

The first upper bounds for partial match was obtained by Rivest
\cite{rivest76pm}, who showed that the trivial $2^d$ space can be
slightly improved when $d \le 2\lg n$. Charikar, Indyk, and
Panigrahy~\cite{charikar02pm} showed that query time $O(n/2^\tau)$ can
be achieved with space $n\cdot 2^{O(d\lg^2 d / \sqrt{\tau/\lg n})}$.
It is generally conjectured that the problem follows from the curse of
dimensionality, in the following sense: there is no constant $\eps>0$,
such that query time $O(n^{1-\eps})$ can be supported with space
$\poly(m)\cdot 2^{O(d^{1-\eps})}$.

If the problem is parameterized by the number of stars $k$, it is
trivial to achieve space $O(n)$ and query time $O(2^k)$ by exploiting
the binary alphabet. In the more interesting case when the alphabet
can be large, Cole, Gottlieb, Lewenstein \cite{cole04match} achieve
space $O(n\lg^k n)$ and time $O(\lg^k n \cdot \lg\lg n)$ for any
constant $k$.

Partial match can be reduced \cite{indyk01linfty} to exact near
neighbor in $\ell_1$ or $\ell_2$, and to $3$-approximate near neighbor
in $\ell_\infty$. This is done by applying the following
transformation to each coordinate of the query: $0\mapsto
-\frac{1}{2}$; $\star \mapsto \frac{1}{2}$; $1 \mapsto \frac{3}{2}$.

\paragraph{Marked ancestor.}
In this problem, defined by Alstrup, Husfeldt, and
Rauhe~\cite{alstrup98marked}, we are to maintain a complete tree of
degree $b$ and depth $d$, in which vertices have a mark bit. The
updates may mark or unmark a vertex. The query is given a leaf $v$,
and must determine whether the path from the root to $v$ contains any
marked node. In our reduction, we work with the version of the problem
in which \emph{edges} are labeled, instead of nodes. However, note
that the problems are identical, because we can attach the label of an
edge to the lower endpoint.

Marked ancestor reduces to dynamic stabbing in 1D, by associating each
vertex with an interval extending from the leftmost to the rightmost
leaf in its subtree. Marking a node adds the interval to the set, and
unmarking removes it. Then, an ancestor of a leaf is marked iff the
leaf stabs an interval currently in the set.

The decremental version, in which we start with a fully marked tree
and may only unmark, can be reduced to union-find. Each time a node is
unmarked, we union it with its parent. Then, a root-to-leaf path
contains no marked nodes iff the root and the leaf are in the same
set.

The lower bounds of this paper work for both the decremental and
incremental variants.

\end{document}

%% file: tikz-butterfly.tex
  \centering
  \begin{tikzpicture}[scale=0.4]
    \foreach \x in { 0, ..., 15 } {
      \fill (\x, 12) circle (0.15);
      \foreach \y in { 0, 3, 6, 9 } {
	\fill (\x, \y) circle (0.15);
	\FPeval{ly}{\y + 0.15}
        \FPeval{uy}{\y + 2.85}
	\draw  (\x, \uy) -- (\x, \ly) ;
      }
    }
    \foreach \x in { 0, ..., 7 } {
      \FPeval{even}{\x * 2}   \FPeval{odd}{\x * 2 + 1}
      \draw (\even, 11.84) -- (\odd, 9.16) ;
      \draw (\odd, 11.84) -- (\even, 9.16) ;
    }
    \foreach \hix in { 0, ..., 3 }
      \foreach \lox in { 0, 1 } {
	\FPeval{even}{\hix * 4 + \lox}   \FPeval{odd}{\hix * 4 + 2 + \lox}
	\draw (\even, 8.84) -- (\odd, 6.16) ;
	\draw (\odd, 8.84) -- (\even, 6.16) ;
      }
    \foreach \hix in { 0, 1 }
      \foreach \lox in { 0, ..., 3 } {
	\FPeval{even}{\hix * 8 + \lox}   \FPeval{odd}{\hix * 8 + 4 + \lox}
	\draw (\even, 5.84) -- (\odd, 3.16) ;
	\draw (\odd, 5.84) -- (\even, 3.16) ;
      }
    \foreach \x in { 0, ..., 7 } {
      \FPeval{even}{\x}   \FPeval{odd}{\x + 8}
      \draw [->] (\even, 2.84) -- (\odd, 0.16) ;
      \draw [->] (\odd, 2.84) -- (\even, 0.16) ;
    }
  \end{tikzpicture}

%% file: tikz-persistent.tex
  \centering
  {\bf (a)}
  \begin{tikzpicture}[scale=0.4]
    \fill (7.5, 12) circle (.15);
    \draw (7.5, 11.84) -- (3.5, 9.16);
    \draw (7.5, 11.84) -- (11.5, 9.16);
    \foreach \i in { 0, 1 } {
      \FPeval{\x}{3.5 + 8 * \i}
      \fill (\x, 9) circle (.15) node[right] {{\small $a_{\i}$}}; 
      \FPeval{\lx}{\x - 2}  \FPeval{\rx}{\x + 2}  
      \draw (\x, 8.84) -- (\lx, 6.16);
      \draw (\x, 8.84) -- (\rx, 6.16);
    }
    \foreach \i in { 0, ..., 3 } {
      \FPeval{\x}{1.5 + 4 * \i}
      \fill (\x, 6) circle (.15) node[right] { {\small $b_{\i}$ }}; 
      \FPeval{\lx}{\x - 1}  \FPeval{\rx}{\x + 1}
      \draw (\x, 5.84) -- (\lx, 3.16);
      \draw (\x, 5.84) -- (\rx, 3.16);
    }
    \foreach \i in { 0, ..., 7 } {
      \FPeval{\x}{0.5 + 2 * \i}
      \fill (\x, 3) circle (.15) node[right] { {\small $c_{\i}$ }}; 
      \FPeval{\lx}{\x - 0.5}  \FPeval{\rx}{\x + 0.5}
      \draw (\x, 2.84) -- (\lx, 0.16);
      \draw (\x, 2.84) -- (\rx, 0.16);
    }
    \foreach \x in { 0, ..., 15 }
      \fill (\x, 0) circle(.15) node[below] { {\small $d_{\x}$}};
  \end{tikzpicture}
\\[2ex]
  {\bf (b)}
  \begin{tikzpicture}[scale=0.4]
    \draw (-0.5, 1) -- (15.5, 1) -- (15.5, 2) -- (-0.5, 2) -- (-0.5, 1);
    \foreach \x in { 0, ..., 15 }
      \draw (\x, 1.5) node { {\small $d_{\x}$}};
    \draw (4, 4) -- (7.5, 2) -- (12, 4);
    \foreach \i in { 0, 1 } {
      \FPeval{\lx}{\i * 8 - 0.4}  \FPeval{\rx}{(\i+1) * 8 - 0.6}
      \draw (\lx, 4) -- (\rx, 4) -- (\rx, 5) -- (\lx, 5) -- (\lx, 4);
      \foreach \x in { 0, ..., 7 } {
	 \FPeval{\nx}{\i * 8 + \x}
         \draw (\nx, 4.5) node { {\small $c_{\x}$}};
      }
      \FPeval{\mx}{(\lx + \rx) / 2} 
      \FPeval{\mxl}{\mx - 2} \FPeval{\mxr}{\mx + 2}
      \draw (\mxl, 7) -- (\mx, 5) -- (\mxr, 7);
    }
    \foreach \i in { 0, ..., 3 } {
      \FPeval{\lx}{\i * 4 - 0.4}  \FPeval{\rx}{(\i+1) * 4 - 0.6}
      \draw (\lx, 7) -- (\rx, 7) -- (\rx, 8) -- (\lx, 8) -- (\lx, 7);
      \foreach \x in { 0, ..., 3 } {
	 \FPeval{\nx}{\i * 4 + \x}
         \draw (\nx, 7.5) node { {\small $b_{\x}$}};
      }
      \FPeval{\mx}{(\lx + \rx) / 2} 
      \FPeval{\mxl}{\mx - 1} \FPeval{\mxr}{\mx + 1}
      \draw (\mxl, 10) -- (\mx, 8) -- (\mxr, 10);
    }
    \foreach \i in { 0, ..., 7 } {
      \FPeval{\lx}{\i * 2 - 0.4}  \FPeval{\rx}{(\i+1) * 2 - 0.6}
      \draw (\lx, 10) -- (\rx, 10) -- (\rx, 11) -- (\lx, 11) -- (\lx, 10);
      \foreach \x in { 0, 1 } {
	 \FPeval{\nx}{\i * 2 + \x}
         \draw (\nx, 10.5) node { {\small $a_{\x}$}};
      }
    }
  \end{tikzpicture}